\documentclass[11pt,letterpaper]{article}

\usepackage{graphicx}
\usepackage{url}
\usepackage{amsfonts}
\usepackage{amsmath}
\usepackage{amsthm}
\usepackage{amssymb}
\usepackage{url}
\usepackage{subfig}
\usepackage[colorlinks=true,citecolor=black,linkcolor=black,urlcolor=blue]{hyperref}
\usepackage{fullpage}
\usepackage{xcolor}
\usepackage[mathlines]{lineno}

\usepackage{easyReview}

\newcommand*\patchAmsMathEnvironmentForLineno[1]{%
  \expandafter\let\csname old#1\expandafter\endcsname\csname #1\endcsname
  \expandafter\let\csname oldend#1\expandafter\endcsname\csname end#1\endcsname
  \renewenvironment{#1}%
     {\linenomath\csname old#1\endcsname}%
     {\csname oldend#1\endcsname\endlinenomath}}%
\newcommand*\patchBothAmsMathEnvironmentsForLineno[1]{%
  \patchAmsMathEnvironmentForLineno{#1}%
  \patchAmsMathEnvironmentForLineno{#1*}}%
\AtBeginDocument{%
\patchBothAmsMathEnvironmentsForLineno{equation}%
\patchBothAmsMathEnvironmentsForLineno{align}%
\patchBothAmsMathEnvironmentsForLineno{flalign}%
\patchBothAmsMathEnvironmentsForLineno{alignat}%
\patchBothAmsMathEnvironmentsForLineno{gather}%
\patchBothAmsMathEnvironmentsForLineno{multline}%
}

\allowdisplaybreaks

\newtheorem{theorem}{Theorem}[section]
\newtheorem{proposition}[theorem]{Proposition}
\newtheorem{lemma}[theorem]{Lemma}

\theoremstyle{remark}

\newtheorem{remark}[theorem]{Remark}

\newcommand{\E}{\ensuremath{\mathbb{E}}}
\newcommand{\pr}{\mathrm{Pr}}

\title{Matching random colored points with rectangles\\
(Corrigendum)}

\author{
Josu\'e Corujo\thanks{
	Univ Paris Est Créteil, Univ Gustave Eiffel, CNRS, LAMA UMR 8050, F-94010 Créteil, France
	 {\tt josue.corujo-rodriguez@u-pec.fr}.
	ORCID: 0000-0002-3997-7391.
}
\and
Paul Horn\thanks{Department of Mathematics, College of Natural Sciences and Mathematics, University of Denver, USA {\tt paul.horn@du.edu}.
	ORCID: 0000-0003-3022-9036.
}
\and
Pablo P\'erez-Lantero\thanks{\emph{Corresponding author}. Universidad de Santiago de Chile (USACH), Facultad de Ciencia, Departamento de Matem\'atica y Ciencia de la Computaci\'on, Chile. {\tt pablo.perez.l@usach.cl}. ORCID: 0000-0002-8703-8970.}
}

\begin{document}

\maketitle

\begin{abstract}
Given $n>0$, let $S\subset [0,1]^2$ be a set of $n$ points, chosen uniformly at random. Let $R\cup B$ be a random partition, or coloring, of $S$ in which each point of $S$ is included in $R$ uniformly at random with probability $1/2$. Corujo et al.~(JOCO 2023) studied the random variable $M(n)$ equal to the number of points of $S$ that are covered by the rectangles of a maximum matching of $S$ using pairwise-disjoint rectangles. Each rectangle is axis-aligned and covers exactly two points of $S$ of the same color. They designed a deterministic algorithm to match points of $S$, and the algorithm was modeled as a discrete stochastic process over a finite set of states. After claiming that the stochastic process is a Markov chain, they proved that almost surely $M(n)\ge 0.83\,n$ for $n$ large enough. The issue is that such a process is not actually a Markov one, as we discuss in this note. We argue this issue, and correct it by obtaining the same result but considering that the stochastic process is not Markov, but satisfies some kind of first-order homogeneity property that allows us to compute its marginal distributions.
\end{abstract}

\section{Introduction}

Let $S\subset [0,1]^2$ be a set of $n$ points, chosen uniformly at random, and let $R\cup B$ be a random coloring of $S$ in which each point of $S$ is included in $R$ uniformly at random with probability $1/2$. A matching of $S$ with rectangles, or simply a {\em matching}, is a set of pairwise-disjoint axis-aligned rectangles where each rectangle covers exactly two points of $S$ of {\em the same color}. Let $M(n)$ be the random variable equal to the number of points of $S$ that are covered by a maximum matching.

Corujo et al.~\cite{corujo2023matching} introduced the study of $M(n)$: They designed a deterministic algorithm to match points of $S$, and the algorithm was modeled as a discrete stochastic process $X_1,X_2,\ldots,X_n$ over a set of $18$ states $\mathcal{E}=\{e_1,e_2,\ldots,e_{18}\}$. The algorithm matches points of $S$ from left to right (in the $x$-coordinate order), and $X_t$ stands for the state of the algorithm after processing the first $t$ points. They defined the function $f:\mathcal{E}\rightarrow\mathbb{R}_{\ge 0}$ with $f(e_i)=2$ for $i\in\{2, 6, 9, 10, 16, 17,18\}$, $f(e_{11})=4$, and $f(e_i)=0$ otherwise; so that the number of points matched by the algorithm equals the random sum $f(X_1)+f(X_2)+\ldots+f(X_n)$. By claiming that the stochastic process is a Markov chain, they computed the stochastic transition matrix $P$, the vector $\overline{s}=(s_1,s_2,\ldots,s_{18})$ which is the stationary distribution of the chain and equal to the unique solution of the system $\overline{s}=\overline{s}\cdot P$, and define $\alpha_3=\sum_{i=1}^{18}f(e_i)\,s_i\approx 0.830030151$. The Ergodic Theorem~\cite[Thm.\ 1.10.2]{norris1998markov} was used to ensure that
$\lim_{n\rightarrow\infty}\frac{1}{n}\sum_{j=1}^n f(X_j) = \sum_{i=1}^{18} s_if(e_i) = \alpha_3$ almost surely, which implies $M(n)\ge 0.83\,n$ almost surely for $n$ large enough.

The issue in the paper of Corujo et al.~\cite{corujo2023matching} is that they claimed without a proper proof that
\begin{equation}\label{eq1}
	\pr(X_{t}=x_{t}\mid X_{t-1}=x_{t-1},\ldots,X_2=x_2,X_1=x_1)=\pr(X_{t}=x_{t}\mid X_{t-1}=x_{t-1}),
\end{equation}
for all $x_1,x_2,\ldots,x_{t}\in\mathcal{E}$ and $t>1$, giving no proof that the stochastic process is a Markov chain. In fact, the probabilities $\pr(X_{t}=e_i\mid X_{t-1}=e_j)$, which are the entries of the transition matrix of the stochastic process $X_1,X_2,\ldots,X_n$ and correspond to the right part of equation~\eqref{eq1}, are conditioned in any path of reaching the state $e_j$ at time $t-1$ (i.e.\ the very recent, or first order past), whereas the left part of equation~\eqref{eq1} is the probability conditioned on precisely one path (i.e.\ the complete history).

Concretely, one can check that the process defined by Corujo et al.~\cite{corujo2023matching} is not a Markov chain by comparing, for instance, the conditional probabilities
 \(
 	\Pr ( X_t = e_4 \mid X_{t-1} = e_4 )
 \)
 and
 \[
 	\Pr ( X_t = e_4 \mid X_{t-1} = e_4, X_{t-2} = e_4, \dots , X_3 = e_4, X_2 = e_3, X_1 = e_1).
 \]
The first probability stands for the event in which the last 3 points, of the first $t$ points seen by the algorithm, are not matched, form a monotone chain and their colors alternate, with the condition that the last 3 points of the first $t-1$ ones satisfy the same property (see Figure~\ref{fig:A}). The second probability stands for the same event but conditioned in that the first $t-1$ points are not matched and form a monotone chain, with color alternation (see Figure~\ref{fig:B}). 

\begin{figure}[t]
	\centering
	\subfloat[]{
		\includegraphics[scale=0.85,page=1]{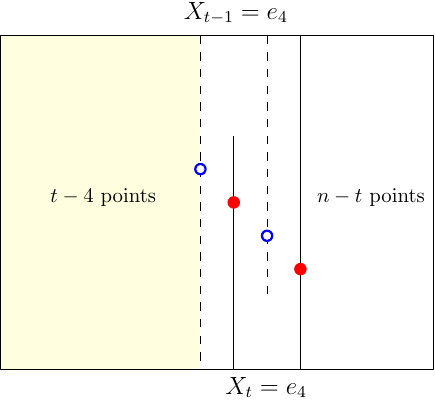}\label{fig:A}
	}~~~~~~
	\subfloat[]{
		\includegraphics[scale=0.85,page=2]{img.pdf}\label{fig:B}
	}
	\caption{\small{(a) The last 3 points, of the first $t$ points processed by the algorithm of Corujo et al.~\cite{corujo2023matching}, are not matched and form a monotone sequence with color alternation. It also happens for the former 3 last points. Note that the first $t-4$ points can be in any position with respect to the last 4 points, and could be matched between them. (b) All the first $t$ points form a monotone sequence with color alternation.}}\label{fig:Z}
\end{figure}

On the one hand, as proved by Corujo et al.~\cite{corujo2023matching}, we have that
\[
 \Pr ( X_t = e_4 \mid X_{t-1} = e_4 ) ~=~ \frac{1}{8}.
\]
Notice that, although the process $(X_t)$ is not Markov, the one step transition probabilities $\Pr ( X_t = y \mid X_{t-1} = x )$ provided in the above-mentioned paper are correct. On the other hand, for every $t \ge 3$ we have that
\[
	\Pr ( X_t = e_4, X_{t-1} = e_4, \dots , X_3 = e_4, X_2 = e_3, X_1 = e_1) ~=~ \left(\frac{2}{t!}\right)\cdot\left(\frac{1}{2^{t-1}}\right) ~=~\frac{1}{t! \, 2^{t - 2}},
\]
where $2/t!$ is the probability of the relative positions of the $t$ points (either an increasing chain or a decreasing chain), and $1/2^{t-1}$ is the probability of the color alternation.
Hence,
\[
\Pr ( X_t = e_4 \mid X_{t-1} = e_4, X_{k-2} = e_4, \dots , X_3 = e_4, X_2 = e_3, X_1 = e_1) ~=~ \frac{1}{2 t}.
\]
Consequently, for every $t \ge 5$ we obtain
 \[
\Pr ( X_t = e_4 \mid X_{t-1} = e_4, \dots , X_3 = e_4, X_2 = e_3, X_1 = e_1) ~<~ \Pr ( X_t = e_4 \mid X_{t-1} = e_4 ).
\]

In this note, we correct this issue by obtaining the same result without assuming the process is a Markov chain. For a related work, refer to the paper of Farag{\'o}~\cite{farago2022approximating} where general discrete stochastic processes satisfying a first-order homogeneity property are approximated by Markov chains.
Our main result is the following:
\begin{theorem}\label{thm:main}
	Let $n > 0$ be an integer, and let $S \subset [0, 1]^2$ be a set of $n$ points, chosen uniformly at random.
	Let $R \cup B$ be a random partition (i.e., coloring) of $S$ in which each point of $S$ is included in $R$ uniformly at random with probability $1/2$. 
	Hence, there exist some $\beta > 0$ such that
	\begin{equation}\label{eq:main_result}
	\lim_{n \to \infty} \frac{M(n)}{n} ~=~ \lim_{n \to \infty} \frac{\E[M(n)]}{n} ~=~ \sup_{n \in \mathbb{N}} \frac{ \E[M(n)] }{n} ~\ge~ 0.83+\beta, \text{ a.s.}
	\end{equation}
\end{theorem}

\begin{remark}[Connection to \cite{corujo2023matching}]
	The lower bound for \( M(n)/n \) matches the one given in \cite[Thm.\ 2]{corujo2023matching}, and in this regard, we correct its proof and achieve the same result.
	However, the almost sure convergence of \( M(n)/n \) to a deterministic constant is a novel finding.
\end{remark}

The rest of this note is organized as follows.
In Section~\ref{sec:expectation-base}, we consider a generic stochastic process $X_1,X_2,X_3,\ldots$ and provide bounds for the expectation $\E[f(X_1)+f(X_2)+\ldots+f(X_n)]$, where $f$ is a bounded non-negative function. In Section~\ref{sec:expectation}, we use the result of Section~\ref{sec:expectation-base} to prove that $\E[M(n)]\ge (0.83+\beta)n$ for $n$ large enough, where $\beta>0$ is a  constant. In Section~\ref{sec:concentration}, by modeling the colored $S\subset[0,1]^2$ as a random $2n$-dimensional vector in $[0,1]^n\times\{0,1\}^n$ and using the McDiarmid's Inequality~\cite{doob1940regularity}, we first prove that $\pr(M(n)\ge 0.83\,n)\rightarrow 1$ when $n\rightarrow\infty$. After that, we strengthen this statement by showing that $M(n)/n$ converges almost surely to the constant $\sup_n\E[M(n)]/n>0.83$, when $n \to \infty$, concluding the proof of Theorem \ref{thm:main}.

\section{Expectation of the accumulative sum of a discrete stochastic process}\label{sec:expectation-base}

Let $\mathcal{E}=\{e_1,e_2,\ldots,e_N\}$ be a finite set of $N$ states, and let $f:\mathcal{E}\rightarrow\mathbb{R}_{\ge 0}$ be a bounded non-negative function. Let $\mathcal{X}=X_1,X_2,X_3,\ldots$ be a stochastic process in which $X_t\in\mathcal{E}$ for all $t\ge 1$. The variable $X_1$ takes value with probability distribution vector $p_1$, and the $N\times N$ stochastic matrix $P$ represents the transition probabilities of $\mathcal{X}$. That is, the columns of $P$ are probability vectors, and it holds that $P_{i,j}=\pr(X_{t}=e_i\mid X_{t-1}=e_j)$ for all $t>1$ and $i,j\in\{1,2,\ldots,N\}$.
Note that the entries of the transition probability matrix $P$ associated to $\mathcal{X}$ do not depend on $t$.
This property was called first-order homogeneity by Faragó \cite{farago2022approximating} and it is worth noting that it holds whenever the process is Markov but it does not implies the Markov property.

We assume that $P$ is irreducible and aperiodic in the usual sense it is defined by for the transition probabilitiy matrix of Markov chains (cf.\ \cite[Def.\ 3]{farago2022approximating}).
Let the vector $\overline{s}=(s_1,s_2,\ldots,s_n)$ be the unique solution, also a probability vector, of the linear system $\overline{s}=\overline{s}\cdot P$, which exists given that $P$ is a stochastic matrix. Let $\alpha=\sum_{i=1}^Nf(e_i)\,s_i= \overline{f}\cdot \overline{s}$, where $\overline{f}$ is the vector $(f(e_1),f(e_2),\ldots,f(e_N))$.

\begin{lemma}\label{lem:expectation}
For any constant $\varepsilon>0$, it holds that
\[
	n(\alpha-\varepsilon)-O_{\varepsilon}(1) ~\le~ \E[f(X_1)+\ldots+f(X_n)] ~\le~ n(\alpha+\varepsilon)+O_{\varepsilon}(1).
\]
for $n$ large enough, where both terms $O_{\varepsilon}(1)$ stand for constants depending on $\varepsilon$.
\end{lemma}

\begin{proof}
Let $p_t$ be the vector probability of $X_t$. That is, $p_t=(\pr(X_t=e_i))_{i=1}^N$. For $t>1$, we have that $\pr(X_t=e_i)=\sum_{j=1}^N\pr(X_{t-1}=e_j)\cdot \pr(X_t=e_i\mid X_{t-1}=e_j)$, which is equivalent to stating that $p_t$ is the scalar product between the vector $p_{t-1}$ and the $i$-th column of the transpose matrix $Q=P^T$. Hence, this is standard to state that $p_t=p_{t-1}\times Q$, so $p_t=p_1\times Q^{t-1}$ for all $t\ge 1$.

We have that 
\[
	\E[f(X_t)]=\sum_{i=1}^Nf(e_i)\cdot \pr(X_t=e_i)=\overline{f}\cdot p_t=\overline{f}\cdot (p_1\times Q^{t-1}),
\]
and then
\[
	\E[f(X_1)+\ldots+f(X_n)] ~=~ \sum_{t=1}^{n}\E[f(X_t)] ~=~ \overline{f}\cdot (p_1 \times (Q^0+Q^1+Q^2+\ldots+Q^{n-1})).
\]
Given that $P$ is irreducible and aperiodic, it is well-known that the limit $\tilde{Q}=\lim_{n\rightarrow\infty} Q^n$ exists and it is the matrix in which each row is the vector $\overline{s}=(s_1,s_2,\ldots,s_n)$. Let $\delta=\delta(\varepsilon,f)>0$ be a constant that will be specified later, and let $n_0$ be a constant such that $\|Q^t-\tilde{Q}\|_{\infty}<\delta$ for $t\ge n_0$. This implies that $|Q^t_{i,j}-s_j|<\delta$, for all $i,j\in\{1,2,\ldots,N\}$ and $t\ge n_0$. Let then $\tilde{Q}_{-\delta}$ denote the matrix obtained by subtracting $\delta$ to each component of $\tilde{Q}$, and $\tilde{Q}_{+\delta}$ the matrix obtained by adding $\delta$ to each component of $\tilde{Q}$. Let $m=\max_{i=1}^N f(e_i)$ which is constant and finite given that $f$ is bounded.

For $n$ large enough, that is, for $n>n_0$, we obtain the following inequalities, given that the function $f$ is non-negative, the vector $p_1$ is a probability vector, and the matrices $Q^0,Q^1,\ldots,Q^{n-1}$ are all non-negative:
\begin{align*}
	\E[f(X_1)+\ldots+f(X_n)] 
		& ~\ge~ \overline{f}\cdot (p_1 \times (Q^{n_0}+Q^{n_0+1}+\ldots+Q^{n-1})) \\
	 	& ~\ge~ \overline{f}\cdot (p_1 \times ((n-n_0)\,\tilde{Q}_{-\delta})) \\
		& ~=~ (n-n_0)\,(\overline{f}\cdot (p_1 \times \tilde{Q}_{-\delta})) \\
		& ~=~ (n-n_0)\,(\overline{f}\cdot (s_1-\delta,s_2-\delta,\ldots,s_N-\delta)) \\
		& ~=~ (n-n_0)\,(\alpha-\delta(f(e_1)+f(e_2)+\ldots+f(e_N))) \\
		& ~\ge~ (n-n_0)\,(\alpha-\delta\,Nm).
\end{align*}
Let $U$ be the $N\times N$ matrix with all components equal to 1, and $u$ a row vector of $U$. Similarly, we have
\begin{align*}
	\E[f(X_1)+\ldots+f(X_n)] 
		& ~\le~ \overline{f}\cdot (p_1 \times (n_0\, U + Q^{n_0}+Q^{n_0+1}+\ldots+Q^{n-1})) \\
	 	& ~\le~ \overline{f}\cdot (p_1 \times (n_0\,U+(n-n_0)\,\tilde{Q}_{+\delta})) \\
		& ~=~ \overline{f}\cdot (n_0\, p_1 \times U+(n-n_0)\,p_1\times \tilde{Q}_{+\delta}) \\
		& ~=~ \overline{f}\cdot (n_0\, u +(n-n_0)\,(s_1+\delta,s_2+\delta,\ldots,s_N+\delta)) \\
		& ~=~ n_0\,\overline{f}\cdot u + (n-n_0)\,\overline{f}\cdot (s_1+\delta,s_2+\delta,\ldots,s_N+\delta) \\
		& ~\le~ n_0\,Nm + (n-n_0)(\alpha+\delta\,N\,m).
\end{align*}
By taking $\delta=\varepsilon/(Nm)$, we have that
\[
	n(\alpha-\varepsilon)-O_{\varepsilon}(1) ~\le~ \E[f(X_1)+\ldots+f(X_n)] ~\le~ n(\alpha+\varepsilon)+O_{\varepsilon}(1).
\]
The lemma thus follows.
\end{proof}  

\section{Bounding the expectation of $M(n)$}\label{sec:expectation}

Consider the same deterministic algorithm designed by Corujo et al.~\cite{corujo2023matching} to match points of $S$, modeled as a discrete stochastic process $X_1,X_2,\ldots,X_n$ over the set $\mathcal{E}=\{e_1,e_2,\ldots,e_{18}\}$ of $N=18$ states. Consider the same function $f:\mathcal{E}\rightarrow\mathbb{R}_{\ge 0}$ with $f(e_i)=2$ for $i\in\{2, 6, 9, 10, 16,$ $17,18\}$, $f(e_{11})=4$, and $f(e_i)=0$ otherwise; so that the number of points matched by the algorithm is the random sum $f(X_1)+f(X_2)+\ldots+f(X_n)$. From the stochastic transition matrix of $X_1,X_2,\ldots,X_n$, which is irreducible and aperiodic, compute the vector $\overline{s}=(s_1,s_2,\ldots,s_{18})$, and define $\alpha = \sum_{i=1}^{18}f(e_i)\,s_i\approx 0.830030151$ %. Let $\alpha=\alpha_3$ 
and let $\beta=(\alpha - 0.83)/3$. Using Lemma~\ref{lem:expectation} for $\varepsilon=\beta$, we have for $n$ large enough that
\[
	\frac{1}{n}\,\E\left[f(X_1)+\ldots+f(X_n)\right] 
	 ~\ge~  (\alpha-\beta)-\frac{O_{\beta}(1)}{n} ~=~  (0.83+2\beta)-\frac{O_{\beta}(1)}{n} ~\ge~ 0.83+\beta,
\]
which implies that 
\begin{lemma}[Lower bound for the expectation] \label{lemma:expec_bound}
	For $n$ large enough, $\E[M(n)]\ge (0.83+\beta)\,n$.
\end{lemma}

\section{Concentration of $M(n)$}\label{sec:concentration}

A function $F:\mathcal{X}_1\times \mathcal{X}_2\times \ldots\times\mathcal{X}_n\rightarrow\mathbb{R}$ satisfies the {\em bounded differences property} if there exist constants $d_1,d_2,\ldots,d_n$ such that 
\[
	|F(x_1,\ldots,x_{i-1},x_{i},x_{i+1},\ldots,x_n)-F(x_1,\ldots,x_{i-1},x'_{i},x_{i+1},\ldots,x_n)| ~\le~ d_i,
\]
for all $i\in\{1,2,\ldots,n\}$, $(x_1,x_2,\ldots,x_n)\in \mathcal{X}_1\times \mathcal{X}_2\times \ldots\times\mathcal{X}_n$, and $x'_{i}\in\mathcal{X}_i$.

\begin{lemma}[McDiarmid's Inequality~\cite{doob1940regularity}]\label{lem:mcdiarmid}
Let $F:\mathcal{X}_1\times \mathcal{X}_2\times \ldots\times\mathcal{X}_n\rightarrow\mathbb{R}$ satisfy the bounded differences property with constants $d_1,d_2,\ldots,d_n$. Consider independent random variables $X_1,X_2,\ldots,X_n$ where $X_i\in\mathcal{X}_i$ for all $i\in\{1,2,\ldots,n\}$. Then, for any $\varepsilon>0$ we have that
\[
	\text{(i)}~~ \Pr \Big( F(X_1,X_2,\ldots,X_n)-\E[F(X_1,X_2,\ldots,X_n)]\ge \varepsilon \Big) ~\le~ \exp\left(-\frac{2\varepsilon^2}{\sum_{i=1}^nd^2_i}\right),
\]
and
\[
\text{(ii)}~~ \Pr \Big( F(X_1,X_2,\ldots,X_n)-\E[F(X_1,X_2,\ldots,X_n)]\le -\varepsilon \Big) ~\le~ \exp\left(-\frac{2\varepsilon^2}{\sum_{i=1}^nd^2_i}\right).
\]
As an immediate consequence,
\[
	\text{(iii)}~~ \Pr \Big( | F(X_1,X_2,\ldots,X_n)-\E[F(X_1,X_2,\ldots,X_n)] | \ge \varepsilon \Big) ~\le~ 2 \exp\left(-\frac{2\varepsilon^2}{\sum_{i=1}^nd^2_i}\right).
\]
\end{lemma}

The random colored $n$-point set $S\subset [0,1]^2$ can be modeled as a random $2n$-dimensional vector in $[0,1]^n\times \{0,1\}^n$: $S$ is the vector $(Y_1,Y_2,\ldots,Y_n,C_1,C_2,\ldots,C_n)$, where $Y_1,Y_2,\ldots,Y_n$ are i.i.d.\ Uniform([0,1]) random variables, and $C_1,C_2,\ldots,C_n$ are i.i.d.\ Bernoulli($0.5$) random variables. The random point set is then $\{p_i:=(i/n,Y_i)\mid i=1,2,\ldots,n\}$ and $p_i$ is colored red if and only if $C_i=0$.

Let $F(Y_1,\ldots,Y_n,C_1,\ldots,C_n)$ be defined as the number of points of $S=(Y_1,\ldots,Y_n,C_1,\ldots,C_n)$ that are covered by the rectangles of a maximum matching of $S$, divided by $n$. That is, $M(n)/n$.

\begin{proposition}\label{prop:1}
The function $F:[0,1]^n\times \{0,1\}^n\rightarrow[0,1]$ satisfies the bounded differences property with constants $d_1=d_2=\ldots=d_n=4/n$ and $d_{n+1}=d_{n+2}=\ldots=d_{2n}=2/n$.
\end{proposition}

\begin{proof}
For a colored point set $S$, let $\mathcal{R}(S)$ denote the rectangle set of a maximum matching of $S$. Let $S=(y_1,\ldots,y_n,c_1,\ldots,c_n)$ be a precise $n$-point set, $i\in\{1,2,\ldots,n\}$ and $S'$ the point set that is obtained by moving the $i$-th point $p_i=(i/n,y_i)$ vertically, such that its new $y$-coordinate is $y'_i$. Note that $|\mathcal{R}(S')|\ge |\mathcal{R}(S)|-2$. Indeed, if we remove from $\mathcal{R}(S)$ both the rectangle that matches $p_i$ with other point (if it exists), and the unique rectangle that contains the new point $(i/n,y'_i)$ (if it exists), then we will obtain a feasible matching of $S'$ with at least $|\mathcal{R}(S)|-2$ rectangles (i.e.\ at most 4 less points are matched). Doing this analysis in the contrary direction, going from $S'$ to $S$, we have that $|\mathcal{R}(S)|\ge |\mathcal{R}(S')|-2$. We then have
\[
|F(y_1,\ldots,y_{i-1},y_{i},y_{i+1},\ldots,y_n,c_1,\ldots,c_n)-F(y_1,\ldots,y_{i-1},y'_{i},y_{i+1},\ldots,y_n,c_1,\ldots,c_n)| ~\le~ 4/n,
\]
for all $i\in\{1,2,\ldots,n\}$. Similarly, if what we change is the color of $c_i$ of $p_i$ to $c'_i$, obtaining $S'$ from $S$, we will have that $|\mathcal{R}(S')|\ge |\mathcal{R}(S)|-1$ since at most one rectangle is lost. Then,
\[
|F(y_1,\ldots,y_n,c_1,\ldots,c_{i-1},c_{i},c_{i+1},\ldots,c_n)-F(y_1,\ldots,y_n,c_1,\ldots,c_{i-1},c'_{i},c_{i+1},\ldots,c_n)| ~\le~ 2/n,
\]
for  all $i\in\{1,2,\ldots,n\}$. The proposition follows.
\end{proof}

Using Proposition~\ref{prop:1} and Lemma~\ref{lem:mcdiarmid} (Equation~(ii)), we have for any $\varepsilon>0$ that
\begin{align*}
	\Pr\left(\frac{M(n)}{n}\ge\E\left[\frac{M(n)}{n}\right]-\varepsilon\right) 
		& ~\ge~ \Pr\left(\frac{M(n)}{n}>\E\left[\frac{M(n)}{n}\right]-\varepsilon\right) \\
		& ~=~ 1 - \Pr\left(\frac{M(n)}{n}-\E\left[\frac{M(n)}{n}\right]\le-\varepsilon\right) \\
		& ~\ge~ 1 - \exp\left(-\frac{2\varepsilon^2}{\sum_{i=1}^{2n}d^2_i}\right) \\
		& ~=~ 1 - \exp\left(-\frac{2\varepsilon^2}{\sum_{i=1}^{n}(4/n)^2+\sum_{i=n+1}^{2n}(2/n)^2}\right) \\
		& ~=~ 1 - \exp\left(-\frac{2\varepsilon^2}{16/n+4/n}\right) \\
		& ~=~ 1 - \exp\left(-\frac{\varepsilon^2n}{10}\right).
\end{align*}
That is, $\lim_{n\rightarrow\infty} \Pr(M(n)\ge \E[M(n)]-n\,\varepsilon)= 1$ for any $\varepsilon>0$. 

By taking $\varepsilon=\beta$ and using that $\E[M(n)]\ge (0.83+\beta)\,n$ (see Section~\ref{sec:expectation}), we obtain that
\begin{equation}\label{eq2}
\lim_{n\rightarrow\infty} \Pr(M(n)\ge 0.83\,n)= 1.
\end{equation}

We strengthen now Equation~\eqref{eq2}, by using Fekete's Lemma~\cite{fekete1923} and Borel-Cantelli Lemma \cite[Thm.\ 10.10]{Folland1999}.%~\cite{borel,cantelli}.

\begin{lemma}[Existence of the limit] \label{lemma:limit_exists}
	The following limit exists
	\[
		\lim_{n \to \infty} \frac{ \E[M(n)] }{n}. 
	\]
	Moreover,
	\[
	\lim_{n \to \infty} \frac{\E[M(n)]}{n}  = \sup_{n \in \mathbb{N}} \frac{ \E[M(n)] }{n}.
	\]
\end{lemma}

\begin{proof}
	A sequence  $\big(a_n \big)_n$ is {\em superadditive} if $a_{n+m}\ge a_n+a_m$ for all $n, m \in \mathbb{N}$. The Fekete's Lemma~\cite{fekete1923} states that if $\big(a_n \big)_n$ is superadditive, then $\lim_{n \to \infty} \frac{ a_n}{n}=\sup_{n \in \mathbb{N}} \frac{ a_n}{n}$.
	
	Let $S$ be a random set of $n + m$ colored points.
	We can build a maximum matching of the first $n$ left-to-right points of $S$, and a maximum matching of the other $m$ points of $S$.
	Let us denote by $M(n)$ and $\tilde{M}(m)$ the number of points matched by each case, respectively.
	Note that $M(m)$ and $\tilde{M}(m)$ follows the same law.
	Besides, the union of both matchings gives us a feasible matching of $S$ covering exactly $M(n) + \tilde{M}(m)$ points.
	Hence, $M(n + m) \ge M(n) + \tilde{M}(m)$, given that $M(n+m)$ is the optimum. 
	Applying the expectation in both terms of the inequality, we have that
	\[
		\E[M(n + m)] \ge \E[M(n) + \tilde{M}(m)] = \E[M(n)] + \E[M(m)].
	\]
	Therefore, the sequence $(\E[M(n)])_n$ is superadditive, and the lemma thus follows.
%	First note that the sequence  $\big( M(n) \big)_n$ is superadditive, meaning that 
%	\(
%		M(n + m) \ge M(n) + M(m).
%	\)
%	for all $n, m \in \mathbb{N}$. 
%	Indeed, let $S$ be a random set of $n + m$ colored points. If we build a maximum matching of the first $n$ left-to-right points of $S$, and a maximum matching of the other $m$ points of $S$, we will obtain a feasible matching of $S$ which covers exactly $M(n)+M(m)$ points. Hence, $M(n + m) \ge M(n) + M(m)$ given that $M(n+m)$ is the optimum. Applying the expectation in both terms of the inequality, we have that $\E[M(n + m)] \ge \E[M(n) + M(m)] = \E[M(n)] + \E[M(m)]$. Therefore, the sequence $(\E[M(n)])_n$ is superadditive, and the lemma thus follows.
\end{proof}

Using Proposition~\ref{prop:1} and Lemma~\ref{lem:mcdiarmid} (Equation~(iii)), we have for any $\varepsilon>0$ that
\begin{equation}\label{eq3}
	\Pr\left( \left| \frac{M(n)}{n} - \E\left[\frac{M(n)}{n}\right] \right| \ge  \varepsilon\right)
 ~\le~ 2 \exp\left(-\frac{2\varepsilon^2}{\sum_{i=1}^{2n}d^2_i}\right) 
	~=~ 2 \exp\left(-\frac{\varepsilon^2n}{10}\right). 
\end{equation}
Using now Equation~\eqref{eq3} and Lemma \ref{lemma:limit_exists}, we can state the next lemma, whose proof is deferred to the appendix, for completeness.
\begin{lemma}\label{lem:borel}
	For any $\varepsilon>0$, we have that
	\[
		\sum_{n = 1}^{\infty} 	\Pr\left( \left| \frac{M(n)}{n} - \sup_{m \in \mathbb{N}} \E\left[\frac{M(m)}{m}\right] \right| >  \varepsilon\right) < \infty.
	\]
\end{lemma}
Let $(X_n)_n$ be a sequence of random variables and $X$ be a constant, such that
\[
	\sum_{n = 1}^{\infty} \Pr \left( |X_n - X| > \varepsilon \right) < \infty
\]
for all $\varepsilon>0$. It is well-known, by the Borel-Cantelli Lemma, that $X_n$ converges to $X$ almost surely as $n\to\infty$. Then, bringing together Lemma~\ref{lemma:expec_bound}, Lemma~\ref{lemma:limit_exists}, and Lemma~\ref{lem:borel}, we easely conclude the proof of Theorem \ref{thm:main}.

{\small

\noindent{\bf Funding}: {\em Pablo P\'erez-Lantero}, Research supported by project DICYT 042332PL Vicerrector\'ia de Investigaci\'on, Desarrollo e Innovaci\'on USACH (Chile).

\noindent{\bf Data and materials availability statement}: Data and materials sharing not applicable to this article as no datasets were generated or analyzed during the current study.

\bibliographystyle{abbrv}
\bibliography{refs}

\begin{thebibliography}{1}

\bibitem{corujo2023matching}
J.~Corujo, D.~Flores-Pe{\~n}aloza, C.~Huemer, P.~P{\'e}rez-Lantero, and
  C.~Seara.
\newblock Matching random colored points with rectangles.
\newblock {\em Journal of Combinatorial Optimization}, 45(2):81, 2023.

\bibitem{doob1940regularity}
J.~L. Doob.
\newblock Regularity properties of certain families of chance variables.
\newblock {\em Transactions of the American Mathematical Society},
  47(3):455--486, 1940.

\bibitem{farago2022approximating}
A.~Farag{\'o}.
\newblock Approximating general discrete stochastic processes by {M}arkov
  chains.
\newblock {\em Journal of Statistics and Computer Science}, 1:135--145, 2022.

\bibitem{fekete1923}
M.~Fekete.
\newblock {\"U}ber die {V}erteilung der {W}urzeln bei gewissen algebraischen
  {G}leichungen mit ganzzahligen {K}oeffizienten.
\newblock {\em Mathematische Zeitschrift}, 17(1):228--249, 1923.

\bibitem{Folland1999}
G.~B. Folland.
\newblock {\em Real analysis. {Modern} techniques and their applications.}
\newblock Pure Appl. Math., Wiley-Intersci. Ser. Texts Monogr. Tracts. New
  York, NY: Wiley, 2nd ed. edition, 1999.

\bibitem{norris1998markov}
J.~R. Norris.
\newblock {\em Markov chains}, volume~2 of {\em Cambridge Series in Statistical
  and Probabilistic Mathematics}.
\newblock Cambridge University Press, Cambridge, 1998.
\newblock Reprint of 1997 original.

\end{thebibliography}

}

\appendix

\section{Proof of Lemma~\ref{lem:borel}}

\begin{proof}
Let $X_n=M(n)/n$, $Y_n=\E[M(n)]/n$, and $L=\sup_n\E[M(n)]/n$. Let $\varepsilon>0$ be a constant. From Equation~\eqref{eq3}, we have that 
\[
	\Pr\left(|X_n-Y_n|\ge \varepsilon/2\right) ~\le~ 2\, r^n, ~\text{where}~ r=\exp\left(-\frac{\varepsilon^2}{40}\right)<1.
\]
From Lemma~\ref{lemma:limit_exists}, we have for $n\ge n_0$, where $n_0\in\mathbb{N}$ is a constant, that 
\[
	\Pr\left(|Y_n-L|<\varepsilon/2\right) ~=~ 1.
\]
For $n\ge n_0$, we obtain that
\begin{align*}
	\Pr(|X_n-L|\le\varepsilon) 
		& ~\ge~ \Pr(|X_n-Y_n|+|Y_n-L|\le\varepsilon) \\
		& ~\ge~ \Pr\left(|X_n-Y_n|<\varepsilon/2,|Y_n-L|<\varepsilon/2 \right) \\
		& ~=~ \Pr\left(|X_n-Y_n|<\varepsilon/2 \right) \\
		& ~\ge~ 1 - 2\, r^n.
\end{align*}
Hence,
\begin{align*}
	\sum_{n=1}^{\infty}\Pr(|X_n-L|>\varepsilon)
	& ~<~ (n_0-1) + \sum_{n=n_0}^{\infty}\Pr(|X_n-L|>\varepsilon) \\
	& ~\le~ (n_0-1) + \sum_{n=n_0}^{\infty} 2\, r^n \\
	& ~=~ (n_0-1)+\frac{2\,r^{n_0}}{1 - r} \\
	& ~<~ \infty.
\end{align*}
The lemma thus follows.
\end{proof}

\end{document}